\documentclass{IEEEtran}
\usepackage{cite}
\usepackage{amsmath,amssymb,amsfonts}
\usepackage{graphicx}
\usepackage{textcomp}
\def\BibTeX{{\rm B\kern-.05em{\sc i\kern-.025em b}\kern-.08em
    T\kern-.1667em\lower.7ex\hbox{E}\kern-.125emX}}
    
\usepackage[table,pdftex,hyperref,svgnames]{xcolor}  %search for svgnames colors to see what colors are alreadu defined
\usepackage{mathrsfs}
\usepackage{multicol}
\usepackage{wrapfig}
\usepackage[utf8]{inputenc}
\usepackage{amsthm}
\usepackage[shortlabels]{enumitem}
\usepackage{here}
\usepackage{multirow}
\usepackage{stmaryrd}
\usepackage{version}

\excludeversion{TAC}
\includeversion{arXiv}

\usepackage{lscape} %landscape table
\usepackage{comment}
\usepackage{listings}
\usepackage{verbatim}
\usepackage{xcolor}
\usepackage{enumitem}
\usepackage{cases}
\usepackage{xspace}

\usepackage{bm}
\usepackage{bbm}
\usepackage{algorithm2e}
\RestyleAlgo{ruled}
\SetAlCapNameFnt{\small}
\SetAlCapFnt{\small}
\usepackage{algpseudocode}
\newcommand{\1}{\mbox{\fontencoding{U}\fontfamily{bbold}\selectfont1}}
\newcommand{\0}{\mbox{\fontencoding{U}\fontfamily{bbold}\selectfont0}}

\DeclareMathAlphabet{\mathcal}{OMS}{cmsy}{m}{n}

\newtheorem{theorem}{Theorem}

\newtheorem{lemma}[theorem]{Lemma}

\usepackage{hyperref}%
\hypersetup{colorlinks,allcolors=black}
\renewcommand{\natural}{{\mathbb{N}}}

\newcommand{\real}{\ensuremath{\mathbb{R}}}

\DeclareSymbolFont{bbold}{U}{bbold}{m}{n}
\DeclareSymbolFontAlphabet{\mathbbold}{bbold}
\newcommand{\Prob}{\mathbb{P}}

\newcommand{\tauvec}{\pmb{\tau}}
\newcommand{\pivec}{\pmb{\pi}}

\newcommand{\pvec}{\bm{p}}
\newcommand{\qvec}{\bm{q}}

\newcommand{\me}{\mathrm{e}}
\newcommand{\e}{\bm{e}}

\begin{document}

\title{A Stochastic Surveillance Stackelberg Game: \\ Co-Optimizing Defense Placement and Patrol Strategy$^{*}$}

\author{Yohan John$^{1}$, Gilberto D\'iaz-Garc\'ia$^{1}$, Xiaoming Duan$^{2}$, Jason R. Marden$^{1}$, and Francesco Bullo$^{1}$% <-this % stops a space
\thanks{$^{*}$This work was supported by Air Force Office of Scientific Research under Grant FA9550-15-1-0138. See https://arxiv.org/abs/2308.14714 for an extended version.}% <-this % stops a space
\thanks{$^{1}$The authors are with the Center for Control, Dynamical Systems, and Computation, UC Santa Barbara, Santa Barbara, CA 93106-5070, USA. Email: \{{\tt\small yohanjohn@ucsb.edu}, {\tt\small gdiaz-garcia@ucsb.edu}, {\tt\small jrmarden@ucsb.edu}, {\tt\small bullo@ucsb.edu}\}}% 
\thanks{$^{2}$Xiaoming Duan is with the Department of Automation, Shanghai Jiao Tong University, Shanghai, China 200240. Email: \{{\tt\small xduan@sjtu.edu.cn}\}}%
}

\maketitle

%%%%%%%%%%%%%%%%%%%%%%%%%%%%%%%%%%%%%%%%%%%%%%%%%%%%%%%%%%%%%%%%%%%%%%%%%%%%%%%%

\begin{abstract}
Stochastic patrol routing is known to be advantageous in adversarial settings; however, the optimal choice of stochastic routing strategy is dependent on a model of the adversary. We adopt a worst-case omniscient adversary model from the literature and extend the formulation to accommodate heterogeneous defenses at the various nodes of the graph. Introducing this heterogeneity leads to interesting new patrol strategies. We identify efficient methods for computing these strategies in certain classes of graphs. We assess the effectiveness of these strategies via comparison to an upper bound on the value of the game. Finally, we leverage the heterogeneous defense formulation to develop novel defense placement algorithms that complement the patrol strategies.
\end{abstract}

\begin{IEEEkeywords}
Game theory, Markov process, optimization, surveillance.
\end{IEEEkeywords}

%%%%%%%%%%%%%%%%%%%%%%%%%%%%%%%%%%%%%%%%%%%%%%%%%%%%%%%%%%%%%%%%%%%%%%%%%%%%%%%%

\section{Introduction}

\subsubsection*{Problem Description \& Motivation}
Protection of strategic locations such as borders and critical infrastructure requires physical agents to patrol. However, in the presence of sophisticated attackers, deterministic patrol routes can be ineffective. As a result, so-called stochastic surveillance employs randomized patrol routing to prevent attackers from taking advantage of patterns in the movement of surveillance agents. The question then becomes how to design effective stochastic patrols against a resourceful adversary. This strategic interaction is a natural application for game theory. The Stackelberg game formulation, in particular, has become popular with successful deployments in real-world applications including security at Los Angeles International Airport \cite{JP-MJ-JM-FO-CP-MT-CW-PP-SK:08} and the port of Boston~\cite{ES-BA-RY-MT-CB-JD-BM-GM:12}. In this brief, we design game-theoretically grounded patrol strategies specialized for the topology of the environment and an omniscient attacker.

\subsubsection*{Literature Review}
We provide a brief review here and refer the reader to \cite{NB:22} for a thorough recent survey of the robotic surveillance and patrolling literature. The work of Chevaleyre \cite{YC:04} examining deterministic cyclic and partitioning strategies for patrolling ignited a flurry of interest. One ensuing line of inquiry centers on adversarial settings~\cite{NB-NG-FA:09},~\cite{JG-JB:05}. Markov chains were identified as a natural tool for modeling stochastic patrol strategies~\cite{NA-SK-GAK:08},~\cite{GC-AS:11}, and graphs are a common choice for modeling patrol areas \cite{NA-VS-GAK-SK:08, NB-NG-FA:12}. Subsequent game-theoretic analyses model the patrolling problem as a strategic game between the attacker and the surveillance agent(s) \cite{CK-MJ-JT-JP-FO-MT:09, PP-JPP-JM-MT-FO-SK:08}. 

Recently, optimization approaches have been proposed to design effective Markov chains for a given graph~\cite{XD-FB:20a, XD-MG-FB:17o, RP-PA-FB:14b}. The review~\cite{XD-FB:20a} recaps several formulations for computing optimal Markov Chains. We adopt the formulation proposed in~\cite{XD-DP-FB:19b}, the only work highlighted as ``of major importance" in the survey~\cite{NB:22}. In~\cite{XD-DP-FB:19b}, the authors formulate a zero-sum Stackelberg game to identify the Markov chain patrol strategy with the highest capture probability for the surveillance agent against an omniscient attacker. Due to the zero-sum assumption, Strong Stackelberg Equilibria are also Nash Equilibria in the simultaneous-move version of the game~\cite{DK-ZY-CK-VC-MT:11}. However, if the motivations of the attacker are well-understood, then explicitly modeling their payoffs is important~\cite{VMB:07}. The authors of~\cite{XD-DP-FB:19b} provide a universal upper bound on the capture probability, optimal strategies for star and line graphs, and a constant factor suboptimal strategy for complete graphs. However, one limitation is the assumption that the attacker needs the same number of time periods to complete an attack on each node in the graph.

\subsubsection*{Contributions}
The focus of this brief is to relax the uniform attack duration assumption in \cite{XD-DP-FB:19b} and account for varying levels of vulnerability across the network. The resulting model induces new patrol strategies and prompts new questions. For example, we view the nodal attack durations as a limited resource and explore the novel co-optimization of patrol strategy and resource allocation. Our contributions to the literature are fourfold: (i) formulation of the stochastic surveillance Stackelberg game with heterogeneous attack duration; (ii) an upper bound on the optimal value of the game; (iii) effective and efficiently computable strategies for complete and complete bipartite graphs; (iv) accompanying optimal and efficiently computable defense placement algorithms for complete and complete bipartite graphs. We also present bounds on the suboptimality of the proposed strategies and prove that the strategy for complete bipartite graphs is optimal for star graphs. 
\begin{arXiv}
    In the Appendix, the same strategy is shown to be constant factor suboptimal for complete bipartite graphs in the setting of \cite{XD-DP-FB:19b}, i.e., with uniform attack duration.
\end{arXiv}

\section{Preliminaries}

% \subsubsection*{Notation}
% \subsection{Notation}
Let $\real$ be the set of real numbers and $\natural_{\geq 2}$ be the set of natural numbers greater than or equal to 2. Let $\1_n,\0_n$ denote column vectors in $\real^n$ with every entry being $1$ and $0$, respectively. $\e_i$ is the $i$-th basis vector. $\textup{diag}(A)$ denotes a diagonal matrix with the same entries on the diagonal as the matrix $A$. $\textup{diag}(\bm{a})$ denotes a diagonal matrix whose diagonal entries are the values from the vector $\bm{a}$. 
% $[A]_{*,j}$ represents the $j$-th column of the matrix $A$.

\subsection{Markov Chains}
Stochastic surveillance strategies can be represented by a first-order discrete-time Markov chain subordinate to a strongly connected digraph $\mathcal{G} = (\mathcal{V},\mathcal{E})$. Nodes in $\mathcal{V}$ represent points of interest, and edges in $\mathcal{E}$ represent paths between them. A graph with $\lvert \mathcal{V} \rvert = n$ nodes leads to an $n$-state Markov chain with a row-stochastic transition matrix $P\in \real^{n\times n}$ where each entry $p_{ij}$ represents the probability of the surveillance agent moving along the edge from node $i$ to node $j$. A Markov chain $P$ is said to be irreducible if an agent following this surveillance strategy  can reach every node from every node. An irreducible Markov chain $P$ has a unique stationary distribution $\pivec \in \real^{n}$ that satisfies $\pivec^\top P = \pivec^\top$ with $\pivec^\top \1_n = 1$ and $\pi_i \geq 0$ for all $i$. The stationary distribution $\pivec$ represents the agent's long-term visit frequency to each node in the graph. 

Let $X_k\in \{1,\ldots,n\}$ be the value of a Markov chain $P$ at time period $k$. The first hitting time $T_{ij} = \min \{k \mid X_0 = i, X_k = j, k \ge 1\}$ is a random variable representing the number of time periods between the agent leaving node $i \in \mathcal{V}$ and its first arrival at node $j \in \mathcal{V}$. Let $F_k$ be the first hitting time probability matrix where the $(i,j)$-th element is the probability that the surveillance agent's first hitting time from node $i$ to node $j$ is $k$ time periods, i.e., $F_k(i,j) = \Prob(T_{ij} = k)$. The $F_k$ matrices satisfy the following recursion where $F_1 = P$ \cite[Ch.~5, Eq.~(2.4)]{EC:13}:
\begin{equation} \label{eq:F_k}
    F_{k+1} = P(F_k - \textup{diag}(F_k)).
\end{equation}
% Note that~\eqref{eq:F_k} has been generalized to accommodate travel times for the edges of the graph in \cite{XD-MG-FB:17o}.

\subsection{Stackelberg Game}
Consider a surveillance agent patrolling a graph according to a given Markov chain $P$. The attacker chooses a single node and remains stationary there. The surveillance agent captures the attacker if she visits the attacker's node within $\tau \in \natural$ time periods. Otherwise, the attacker succeeds. The Stackelberg game formulation is given by the following \cite{XD-DP-FB:19b}:
\begin{equation} \label{opt:Stack_compact}
\begin{aligned}
    \max_{P\in \real^{n \times n}} \quad & \min_{i,j}\{ \Prob(T_{ij} (P) \leq \tau) \} \\
    \text{s.t.} \quad & P \1_n = \1_n, \\
    & p_{ij} \ge 0, \quad \quad \forall(i,j) \in \mathcal{E}, \\
    & p_{ij} = 0, \quad \quad \forall(i,j) \notin \mathcal{E}.
\end{aligned}
\end{equation}
% The inner minimization represents the attacker choosing the pair of nodes $(i,j)$ such that the surveillance agent has the lowest probability of traveling from node $i$ to node $j$ within the attack duration $\tau$. Once the minimizing pair $(i,j)$ is identified, he will attack node $j$ when the surveillance agent arrives at node $i$. 
The inner minimization represents the attacker attacking the node $j$ when the surveillance agent is at node $i$.
This model represents optimal behavior for an omniscient attacker and thus provides a lower bound on the capture probability for any other attacker strategy. The outer maximization reflects the surveillance agent's goal of choosing $P$ to maximize the capture probability. Introducing an auxiliary variable $\mu \in [0,1]$ that represents the capture probability and utilizing~\eqref{eq:F_k}, we can reformulate~\eqref{opt:Stack_compact}~\cite{XD-DP-FB:19b}:
\begin{subequations} \label{opt:Stack}
\begin{align}
    & \max_{\mu \in \real,P \in \real^{n \times n}} \quad \mu \\
    & \text{s.t.} \quad \mu \1_n \1_n^\top \leq \sum\nolimits_{k=1}^\tau F_k, \label{const:mu_XD} \\
    & \quad \quad F_1 = P, \\
    & \quad \quad F_{k+1} = P(F_k - \textup{diag}(F_k)), \quad 1\leq k \leq \tau - 1, \\
    & \quad \quad P \1_n = \1_n, \\
    & \quad \quad p_{ij} \ge 0, \quad \forall(i,j) \in \mathcal{E}, \\
    & \quad \quad p_{ij} = 0, \quad \forall(i,j) \notin \mathcal{E},
\end{align}
\end{subequations}
where the inequality~\eqref{const:mu_XD} is element-wise. An optimal solution $P^*$ to~\eqref{opt:Stack} is irreducible and thus has a unique stationary distribution $\pivec^*$ \cite{XD-DP-FB:19b}. 
 
\section{Heterogeneous Attack Duration}
In this section we extend the formulation in~\cite{XD-DP-FB:19b} to allow the amount of time required for a successful attack to vary for different nodes in the network.
% This could represent varying levels of defense at each location. 
We replace the scalar attack duration $\tau \in \natural$ by a vector $\tauvec \in \natural^n$ where each entry $\tau_i$ represents the attack duration for the $i$-th node. The Stackelberg game~\eqref{opt:Stack} can be modified accordingly:
\begin{subequations} \label{opt:Stack_tauvec}
\begin{align}
    & \max_{\mu \in \real,P \in \real^{n \times n}} \quad \mu \\
    & \text{s.t.} \quad \mu \1_n \1_n^\top \leq \sum\nolimits_{j=1}^n \sum\nolimits_{k=1}^{\tau_j} F_k \e_j \e_j^\top, \label{const:mu} \\
    & \quad \quad F_1 = P, \\
    & \quad \quad F_{k+1} = P(F_k - \textup{diag}(F_k)), \ 1\leq k \leq \tau_\textup{max} {-} 1, \label{const:recursion} \\
    & \quad \quad P \1_n = \1_n, \\
    & \quad \quad p_{ij} \ge 0, \quad \forall(i,j) \in \mathcal{E}, \\
    & \quad \quad p_{ij} = 0, \quad \forall(i,j) \notin \mathcal{E},
\end{align}
\end{subequations}
where $\tau_\textup{max} := \max_{1\leq i \leq n} \tau_i$. The modified RHS of~\eqref{const:mu} yields a matrix such that the $j$-th column is the sum of the $j$-th columns of $F_k$ for $k\in\{1,\ldots,\tau_j\}$. We will show that, in contrast with~\cite{XD-DP-FB:19b} where the proposed strategies $P$ can be written down explicitly as a function of $n$, effective patrol strategies for special cases of~\eqref{opt:Stack_tauvec} must be identified via the solution of nonlinear equation(s) parametrized by $\tauvec$. For nontrivial solutions to~\eqref{opt:Stack_tauvec},  $\tauvec$ must satisfy:
\begin{enumerate}
    \item Each $\tau_i$ must be greater than or equal to the number of ``hops" to the furthest node from node $i$ in the graph $\mathcal{G}$. Otherwise, the capture probability is zero.
    \item At least one $\tau_i$ must be less than the length of every closed cycle on $\mathcal{G}$ that visits all nodes. Otherwise, the surveillance agent can achieve a capture probability of one by using this cycle as a deterministic strategy. 
\end{enumerate}

\subsection{Upper Bound for Optimal Capture Probability}
We begin our analysis of~\eqref{opt:Stack_tauvec} by presenting a general upper bound for the optimal capture probability $\mu^*$. 
\begin{theorem}[Upper Bound for Optimal Capture Probability] \label{thm:up_bound}
Given a strongly connected digraph $\mathcal{G} = (\mathcal{V},\mathcal{E})$ with $\lvert \mathcal{V} \rvert = n$ and a vector $\tauvec \in \natural^n$ of attack durations corresponding to each node in $\mathcal{V}$, the optimal strategy $P^*$ has a unique stationary distribution $\pivec^*$ and corresponding optimal capture probability $\mu^* \leq \min_{1\leq i \leq n} \pi_i^* \tau_i$.
\end{theorem}
\begin{proof}
\begin{TAC}
    The proof closely resembles that of Theorem 5 in \cite{XD-DP-FB:19b} and is available in the arXiv version of this brief~\cite{YJ-GDG-XD-JRM-FB:23}. \qedhere
\end{TAC}

\begin{arXiv}
    Let $(\mu^*,P^*)$ be an optimal solution to~\eqref{opt:Stack_tauvec}. Then $P^*$ is irreducible and has a unique stationary distribution $\pivec^*$. Pre-multiply both sides of~\eqref{const:recursion} by $(\pivec^*)^\top$ to obtain
    \begin{equation} \label{eq:ineq_chain}
    \begin{aligned}
        (\pivec^*)^\top F_{k+1} & = (\pivec^*)^\top (F_k - \textup{diag}(F_k)) \leq (\pivec^*)^\top F_k \\
        & \leq (\pivec^*)^\top F_1 = (\pivec^*)^\top
    \end{aligned} 
    \end{equation}
    using $F_1 = P$, the fact that $P$ is row-stochastic, and the definition of the stationary distribution. Since $(\mu^*,P^*)$ is an optimal solution to~\eqref{opt:Stack_tauvec}, it satisfies~\eqref{const:mu}:
    \begin{equation}
        \mu^* \1_n \1_n^\top \leq \sum_{j=1}^n \sum_{k=1}^{\tau_j} F_k E_j.
    \end{equation}
    Pre-multiply both sides by $(\pivec^*)^\top$ and use the inequality~\eqref{eq:ineq_chain} to obtain
    \begin{equation}
    \begin{aligned}
        \mu^* \1_n^\top & \leq \sum_{j=1}^n \sum_{k=1}^{\tau_j} (\pivec^*)^\top F_k E_j \leq \sum_{j=1}^n \sum_{k=1}^{\tau_j} (\pivec^*)^\top  E_j \\
        & \leq (\pivec^*)^\top \textup{diag}(\tauvec)
    \end{aligned}
    \end{equation}
    which implies the desired result $\mu^* \leq \min_{1\leq i \leq n} \pi_i^* \tau_i$.
\end{arXiv}
\end{proof}

\subsection{Complete Graphs} \label{sect:opt_complete}
In this section, we present a method for computing an effective patrol strategy with bounded suboptimality for complete graphs. Consider the class of strategies $P = \1_n\pivec^\top$, a generalization of the uniform random walk strategy $P=(1/n)\1_n \1_n^\top$ shown in \cite{XD-DP-FB:19b} to be constant factor suboptimal in the uniform attack duration setting. The following lemma gives the capture probability for these strategies.
\begin{lemma}[Capture Probability for Complete Graph Strategy] \label{lem:pi_strat}
Given a complete digraph $\mathcal{G} = (\mathcal{V},\mathcal{E})$ with $|\mathcal{V}|=n$ and a vector $\tauvec \in \natural^n$ of attack durations corresponding to each node in $\mathcal{V}$, the capture probability for the strategy $P = \1_n\pivec^\top$ is $\mu = \min_{1\leq i \leq n} \bigl[ 1 - (1-\pi_i)^{\tau_i} \bigr]$.
\end{lemma}
\begin{proof}
The $i$-th column of the first hitting time probability matrix $F_k \e_i$ corresponding to the strategy $P = \1_n\pivec^\top$ can be calculated using the recursion~\eqref{eq:F_k}:
\begin{equation}
    F_k \e_i = \Bigl[ \pi_i (1-\pi_i)^{k-1} \Bigr] \1_n.
\end{equation}
Then the sum of the $i$-th columns of the $F_k$ matrices for $k \in \{1, \ldots, \tau_i\}$ can be calculated:
\begin{equation}
\begin{aligned}
    \sum\nolimits_{k=1}^{\tau_i} F_k \e_i &= \Bigl[ \pi_i \sum\nolimits_{k=1}^{\tau_i} (1 {-} \pi_i)^{k-1} \Bigr] \1_n \\
    &= \bigl[ 1 - (1 {-} \pi_i)^{\tau_i} \bigr] \1_n.    
\end{aligned}
\end{equation}
Therefore, the capture probability is given by
\begin{equation}
    \mu = \min_{1\leq i \leq n} \bigl[ 1 - (1-\pi_i)^{\tau_i} \bigr]. \qedhere
\end{equation}
\end{proof}
Using Lemma~\ref{lem:pi_strat}, we can write an optimization problem to maximize the capture probability via the choice of the stationary distribution $\pivec$:
\begin{equation} \label{opt:pi_opt}
\begin{aligned}
    \max_{\pivec \in \real^n} \quad & \min_{1\leq i \leq n} -(1-\pi_i)^{\tau_i} \\
    \text{s.t.} \quad & \pivec \geq \0_n, \\
    & \pivec^\top \1_n = 1.
\end{aligned}
\end{equation}
Note that~\eqref{opt:pi_opt} can be viewed as a resource allocation problem with continuous decision variables $\pi_i$ and a separable, strictly increasing, and invertible objective function $f_i(\pi_i) := -(1-\pi_i)^{\tau_i}$. As a result,~\eqref{opt:pi_opt} can be reduced to a single nonlinear equation~\cite{JRB:79}:
\begin{equation} \label{eq:v_opt_calc}
    \sum\nolimits_{i=1}^n f^{-1}_i(v_\textup{opt}) = \sum\nolimits_{i=1}^n \bigl( 1 - (-v_\textup{opt})^{1/\tau_i} \bigr) = 1,
\end{equation}
where $v_\textup{opt} \in (-1,0)$ is the optimal value of~\eqref{opt:pi_opt}. A unique solution to Eq.~\eqref{eq:v_opt_calc} exists for all $\tauvec \in \natural^n$ because the function $g(v_\textup{opt}) := \sum_{i=1}^n (-v_\textup{opt})^{1/\tau_i}$ has a range of $[0,n]$ and is strictly increasing. Eq.~\eqref{eq:v_opt_calc} can be solved efficiently for $v_\textup{opt}$ via bisection. Then the corresponding choice of decision variables $\pivec_\textup{opt}$ can be found via
\begin{equation}\label{eq:v_opt}
    -v_\textup{opt} = (1-\pi_{\textup{opt},1})^{\tau_1} = \cdots = (1-\pi_{\textup{opt},n})^{\tau_n}.
\end{equation}

We can get a lower bound on the suboptimality of the strategy $P_\textup{opt} = \1_n\pivec^\top_\textup{opt}$ by considering the ratio of its capture probability $\mu_\textup{opt}$ to the upper bound for the optimal capture probability $\mu^*$ given by Theorem~\ref{thm:up_bound}:
\begin{equation}
    \frac{\mu_\textup{opt}}{\mu^*} \geq \frac{1+v_\textup{opt}}{\min_{1\leq i \leq n} \pi_i^* \tau_i} \geq 
    \frac{1+v_\textup{opt}}{\min \{ 1, \frac{\tau_\textup{max}}{n} \}},
\end{equation}
where the second inequality comes from the fact that $\mu^*$ is a probability and $\pi^*_\textup{min} \leq 1/n$.

\subsection{Complete Bipartite Graphs} \label{sect:opt_bipart}
In this section, we present a method for computing an effective patrol strategy with bounded suboptimality for complete bipartite graphs. Denote a complete bipartite graph by $\mathcal{G} = (\mathcal{P},\mathcal{Q},\mathcal{E})$ where $\mathcal{P},\mathcal{Q}$ are two sets of disjoint and independent nodes and $\lvert \mathcal{P} \rvert = n_p, \lvert \mathcal{Q} \rvert = n_q$. Motivated by the complete graph results, we consider strategies of the form:
\begin{equation} \label{eq:bipart_strat}
    P = 
    \begin{bmatrix}
        \0_{n_p \times n_p} & \1_{n_p} \qvec^\top \\[0.1cm]
        \1_{n_q} \pvec^\top & \0_{n_q \times n_q}
    \end{bmatrix},
\end{equation}
where $\pvec \in \real^{n_p},\qvec \in \real^{n_q}$ are the transition probabilities to the nodes in $\mathcal{P},\mathcal{Q}$, respectively. The following lemma gives the capture probability for these strategies.
\begin{lemma}[Capture Probability for Complete Bipartite Graph Strategy] \label{lem:bipart_strat}
Given a complete bipartite digraph $\mathcal{G} = (\mathcal{P},\mathcal{Q},\mathcal{E})$ with $\lvert \mathcal{P} \rvert = n_p,\lvert \mathcal{Q} \rvert = n_q, n_p + n_q = n$ and vectors $\tauvec^p \in \natural^{n_p}_{\geq2}, \tauvec^q \in \natural^{n_q}_{\geq2}$ of attack durations corresponding to each node in $\mathcal{P},\mathcal{Q}$, the capture probability for strategies of the form of Eq.~\eqref{eq:bipart_strat} is
\begin{equation} \label{eq:bipart_mu}
\begin{aligned}
    \mu = \min \Bigl\{ & 1{-}(1{-}p_1)^{\lfloor \tau_1^p/2 \rfloor}, \ldots, 1{-}(1{-}p_{n_p})^{\lfloor \tau_{n_p}^p/2 \rfloor}, \\
    & 1{-}(1{-}q_1)^{\lfloor \tau_1^q/2 \rfloor}, \ldots, 1{-}(1{-}q_{n_q})^{\lfloor \tau_{n_q}^q/2 \rfloor} \Bigr\},
\end{aligned}
\end{equation}
where $\pvec = \begin{bmatrix}
        p_1 & \cdots & p_{n_p} 
    \end{bmatrix}^\top$ and $\qvec = \begin{bmatrix}
        q_1 & \cdots & q_{n_q} 
    \end{bmatrix}^\top$.
\end{lemma}
\begin{proof}
Using the recursion~\eqref{eq:F_k}, an equation for the minimum entry of the $i$-th column of $F_k$ can be written for $k \geq 2$:
\begin{equation} \label{eq:min_F_k_bipart}
    \min F_k \e_i  = p_i (1 - p_i)^{\lfloor k/2 \rfloor  -1},
\end{equation}
assuming that node $i \in \mathcal{P}$. The same equation can be written for a node $j \in \mathcal{Q}$. Note that the minimum value of $F_k \e_i$ only increases for successive even values of $k$. Using Eq.~\eqref{eq:min_F_k_bipart}, the sum of the $i$-th columns of the $F_k$ matrices for $k \in \{2, \ldots, \tau_i^p\}$ can be calculated:
\begin{equation}
\begin{aligned}
    \sum\nolimits_{k=2}^{\tau_i^p} F_k \e_i &= p_i \sum\nolimits_{k=2}^{\tau_i^p} (1{-}p_i)^{\lfloor k/2 \rfloor -1} \\ &= 1 - (1{-}p_i)^{\lfloor \tau_i^p/2 \rfloor}.
\end{aligned}
\end{equation}
The capture probability is given by the minimum of $n$ terms like this, one for each column of $P$, as shown in~\eqref{eq:bipart_mu}. \qedhere
\end{proof}
Using Lemma~\ref{lem:bipart_strat}, we can write an optimization problem to identify the optimal $\pvec^*$:
\begin{equation} \label{opt:p_opt}
    \begin{aligned}
        \max_{\pvec \in \real^{n_p}} \quad & \min_{i \in \{1,\ldots,n_p\}} -(1-p_i)^{\lfloor \tau_i^p/2 \rfloor} \\
        \text{s.t.} \quad & \pvec \geq \0_{n_p}, \\
        & \pvec^\top \1_{n_p} = 1,
    \end{aligned}   
\end{equation}
and an identical formulation for $\qvec^*$. The optimal $\pvec^*$ can then be found via solution of the following equations by the argument presented in Section~\ref{sect:opt_complete}:
\begin{equation} \label{eq:v_opt_bipart}
    \begin{aligned}
        & \sum\nolimits_{i=1}^{n_p} \Bigl[ 1 - \bigl(-v_\textup{opt}^p\bigr)^{1/\lfloor \tau_i^p/2 \rfloor} \Bigr] = 1, \\
        & -v_\textup{opt}^p = (1-p_1^*)^{\lfloor \tau_1^p/2 \rfloor} = \cdots = (1-p_{n_p}^*)^{\lfloor \tau_{n_p}^p/2 \rfloor}.        
    \end{aligned}
\end{equation}
The same equations hold for $\qvec^*$. The stationary distribution $\pivec$ of this optimized strategy is given by
\begin{equation}
    \pivec = \begin{bmatrix}
        \frac{q_1^*}{2} & \cdots & \frac{q_{n_q}^*}{2} & \frac{p_1^*}{2} & \cdots & \frac{p_{n_p}^*}{2}
    \end{bmatrix}^\top.
\end{equation}

We can get a lower bound on the suboptimality of this optimized strategy by considering the ratio of its capture probability $\mu_\textup{opt}$ to the upper bound for the optimal capture probability $\mu^*$ given by Theorem~\ref{thm:up_bound}:
\begin{equation}
    \frac{\mu_\textup{opt}}{\mu^*} \geq
    \frac{1 + \min \{v_\textup{opt}^p,v_\textup{opt}^q \}}{\min \{ 1, \frac{\tau_\textup{max}}{n} \}}.
\end{equation}

\subsection{Star Graphs}
In this section, we prove that the strategy outlined in the previous section for complete bipartite graphs is optimal in the special case of star graphs. First, we remark that Lemma 10, Lemma 11, part (i) of Lemma 9, and the first part of the proof of Theorem 12 in \cite{XD-DP-FB:19b} continue to hold in the heterogeneous attack duration setting (simply replace $\tau$ in their proofs with the entry of $\tauvec$ corresponding to the attacker's node). Combining these results with those of the previous section yields the following theorem.
\begin{theorem}[Optimal Strategy for Star Graphs]
    Given a star digraph $\mathcal{G} = (\mathcal{P},\mathcal{Q},\mathcal{E})$ with $\lvert \mathcal{P} \rvert = 1, \lvert \mathcal{Q} \rvert = n-1$ and attack durations $\tauvec \in \natural^n_{\geq 2}$ corresponding to each node in $\mathcal{P},\mathcal{Q}$, the optimal strategy with capture probability $\mu^* = 1 + v_\textup{opt}^p$ for the surveillance agent is
    \begin{equation}
        P^* = \begin{bmatrix}
            0 & q_2^* & \cdots & q_{n}^* \\
            1 & 0 & \cdots & 0 \\
            \vdots & \vdots & \cdots & \vdots \\
            1 & 0 & \cdots & 0
            \end{bmatrix},
    \end{equation}
    where $q_2^*,\ldots,q_{n}^*$ are found via the solution of
    \begin{equation} \label{eq:star_opt}
    \begin{aligned}
        & \sum\nolimits_{i=2}^{n} \Bigl[ 1 - \bigl(-v_\textup{opt}^p\bigr)^{1/\lfloor \tau_i/2 \rfloor} \Bigr] = 1, \\
        & -v_\textup{opt}^p = (1-q_2^*)^{\lfloor \tau_2/2 \rfloor} = \cdots = (1-q_n^*)^{\lfloor \tau_n/2 \rfloor}.        
    \end{aligned}
\
    \end{equation}
\end{theorem}
\begin{proof}
\begin{TAC}
    The proof resembles that of Theorem 12 in \cite{XD-DP-FB:19b} and is available in the arXiv version of this brief~\cite{YJ-GDG-XD-JRM-FB:23}. \qedhere
\end{TAC}

\begin{arXiv}
    Rows 2 to $n$ of $P^*$ follow directly from Lemma 11 of \cite{XD-DP-FB:19b}. The zero in the $(1,1)$ entry of $P^*$ prevents a self-loop at the center node per the first part of the proof of Theorem 12 in \cite{XD-DP-FB:19b}. What remains is to show that the optimal choice of $q_2^*,\ldots,q_n^*$ boils down to~\eqref{opt:p_opt} where $i \in \{2,\ldots,n\}$.

    First, notice that the intruder never attacks the center node because $\tau_1 \ge 2$ so the capture probability for the surveillance agent following $P^*$ is 1 at the center node. Then, following the proof of Theorem 12 in \cite{XD-DP-FB:19b}, notice that for all $\tau_j \in \natural_{\geq 2}$ and $j \in \{2,\ldots,n\}$
    \begin{equation}
    \begin{split}
        \Prob(T_{1j} \le \tau_j) & = q_j + \sum_{k\notin \{1,j\}} q_k \Prob(T_{kj} \le \tau_j - 1) \\
        & = q_j + \sum_{k\notin \{1,j\}} q_k \Prob(T_{1j} \le \tau_j - 2) \\
        & = q_j + (1 - q_j) \Prob(T_{1j} \le \tau_j - 2)
    \end{split}
    \end{equation}
    where $\Prob(T_{1j} \le 1) = \Prob(T_{1j} \le 2) = q_j$. Therefore, the capture probability of this strategy when the intruder attacks the $j$-th node is
    \begin{equation}
        \Prob(T_{1j} \le \tau_j) = 1 - (1 - q_j)^{\lceil \tau_j/2 \rceil}.
    \end{equation}
    The intruder chooses when and where to attack based on 
    \begin{equation}
    \begin{split}
        \min_{i,j \in \{2,\ldots,n\}} & \Prob(T_{ij} \le \tau_j) = \min_{i,j \in \{2,\ldots,n\}} \Prob(T_{1j} \le \tau_j - 1) \\ 
        & = \min_{j \in \{2,\ldots,n\}} \Bigl[ 1 - (1 - q_j)^{\lceil (\tau_j-1)/2 \rceil} \Bigr].
    \end{split}
    \end{equation}
    Note that $\lceil (\tau_j-1)/2 \rceil = \lfloor \tau_j/2 \rfloor \in \natural$ because $\tau_j \in \natural_{\geq 2}$ for all $j$. This implies that the optimal strategy for the surveillance agent can be found via
    \begin{equation}
    \begin{aligned}
        \max_{\qvec} \quad & \min_{j \in \{2,\ldots,n\}} -(1-q_j)^{\lfloor \tau_j/2 \rfloor} \\
        \text{s.t.} \quad & q_j \geq 0, \quad \forall j\in\{2,\ldots,n\},\\
        & \qvec^\top \1_{n-1} = 1
    \end{aligned}
    \end{equation}
    which can be solved for $\qvec^*$ using Eq.~\eqref{eq:star_opt} by the argument presented in Section \ref{sect:opt_bipart}.
\end{arXiv}
\end{proof}

The stationary distribution $\pivec^*$ of the optimal strategy $P^*$ is given by
\begin{equation}
    \pivec^* = \begin{bmatrix}
        \frac{1}{2} & \frac{q_{2}^*}{2} & \cdots & \frac{q_{n}^*}{2}
    \end{bmatrix}^\top.
\end{equation}

\section{Optimal Defense Placement}
Interpreting $\tau_i$ as the strength of the defenses at node $i$, we now consider the related problem of optimally allocating a defense budget $B \in \natural$ among the nodes such that $\tauvec^\top \1_n = B$. We leverage the previous results on optimized patrol strategies and design complementary defense placement rules. 

\subsection{Complete Graphs} \label{sect:odp_complete}
In this section, we present the optimal defense placement rule for complete graphs assuming a surveillance strategy of the form $P = \1_n\pivec^\top$. Lemma~\ref{lem:pi_strat} gives the capture probability for these strategies. We can write an optimization problem to maximize the capture probability over the integer attack durations $\tauvec$ with the defense budget constraint:
\begin{equation} \label{eq:def_opt}
\begin{aligned}
    \max_{\tauvec \in \natural^n} \quad & \min_{1\leq i \leq n} -(1-\pi_i)^{\tau_i} \\
    \text{s.t.} \quad & \tauvec \geq \1_n, \\
    \quad & \tauvec^\top \1_n = B.
\end{aligned}
\end{equation}
Note that $n < B < n^2$ for a nontrivial allocation problem. From Section \ref{sect:opt_complete}, we know that the optimal choice of $\pivec$ is given by Eqs.~(\ref{eq:v_opt_calc},\ref{eq:v_opt}). Therefore, we can reformulate~\eqref{eq:def_opt}:
\begin{subequations} \label{eq:def_opt_reform}
\begin{align}
    \min_{w \in \real,\tauvec \in \natural^n} \quad & w \\
    \text{s.t.} \quad & \sum\nolimits_{i=1}^n w^{1/\tau_i} = n-1, \label{const:nonl} \\
    & \tauvec \geq \1_n, \\
    & \tauvec^\top \1_n = B,
\end{align}
\end{subequations}
where we have defined $w := -v_\textup{opt} \in (0,1)$. We arrive at a Mixed-Integer Nonlinear Program (MINLP). Due to the structure of problem~\eqref{eq:def_opt_reform}, we will be able to write down the optimal solution explicitly. First, we will need the lower bound on the optimal value of $w$ given in the following lemma.
\begin{lemma}[Lower Bound for Objective Function Value] \label{lem:w_bound}
    Given a complete digraph $\mathcal{G} = (\mathcal{V},\mathcal{E})$ with $|\mathcal{V}|=n$ and a defense budget $B \in \natural \cap (n,n^2)$, we have the following lower bound for the optimal objective function value of problem~\eqref{eq:def_opt_reform}:
    \begin{equation}
        w^* > \me^{-2}.
    \end{equation}
\end{lemma}

\begin{proof}
    % We proceed by contradiction. 
    Assume $w^* \leq \me^{-2}$. We will show that this cannot satisfy constraint~\eqref{const:nonl} for any feasible choice of $\tauvec$. Note that the function $f(w,\tauvec) = \sum_{i=1}^n w^{1/\tau_i}$ defined by the LHS of constraint~\eqref{const:nonl} is increasing in $w$. Therefore, it suffices to show that $f(\me^{-2},\tauvec) < n-1$ for any feasible $\tauvec$. Now we can show that $f(\me^{-2},\tauvec)$ is increasing and concave w.r.t.\ $\tau_i$:
    \begin{equation} \label{eq:df_dtau}
    \begin{aligned} 
        \frac{\partial f(\me^{-2},\tauvec)}{\partial \tau_i} & = \frac{2\me^{-2/\tau_i}}{\tau_i^2} > 0, \\
        \frac{\partial^2 f(\me^{-2},\tauvec)}{\partial \tau_i^2} & = \frac{-4\me^{-2/\tau_i}(\tau_i - 1)}{\tau_i^4} \leq 0.
    \end{aligned}
    \end{equation}
    Therefore, the following optimization problem is convex:
    \begin{equation} \label{eq:lemma_opt_problem}
    \begin{aligned}
        \max_{\tauvec \in \real^n} \quad & f(\me^{-2},\tauvec) \\
        \text{s.t.} \quad & \tauvec \geq \1_n, \\
        & \tauvec^\top \1_n = B.
    \end{aligned}
    \end{equation}
    Notice that we have relaxed $\tauvec$ to be real-valued. The optimal value of~\eqref{eq:lemma_opt_problem} represents an upper bound on $f(\me^{-2},\tauvec)$ for feasible $\tauvec$. We form the Lagrangian of~\eqref{eq:lemma_opt_problem}:
    \begin{equation}
        \mathcal{L}(\tauvec,\lambda) = f(\me^{-2},\tauvec) + \lambda(\tauvec^\top \1_n - B),
    \end{equation}
    where we have again increased the size of the feasible region by omitting the inequality constraints. Differentiating w.r.t.\ $\tau_i$ yields the following condition for stationary points:
    \begin{equation} \label{eq:stat_pt}
        \frac{\me^{-2/\tau_i}}{\tau_i^2} = \frac{\me^{-2/\tau_j}}{\tau_j^2},
    \end{equation}
    for all $i,j$. The point $\tauvec^* = (B/n)\1_n$ satisfies~\eqref{eq:stat_pt} and the constraints and therefore must be a global maximum of problem~\eqref{eq:lemma_opt_problem}.
    Now we have an upper bound for $f(\me^{-2},\tauvec)$; namely, $f(\me^{-2},\tauvec) \leq f(\me^{-2},\tauvec^*) = n\me^{-2n/B}$. Because $B < n^2$, $n\me^{-2n/B} < n\me^{-2/n}$. It can be shown that $n\me^{-2/n} < n-1$ for $n \geq 2$ which is a contradiction.\qedhere
\end{proof}

This lemma provides an upper bound on the capture probability $\mu_{\textup{opt}} < 1-\me^{-2}$. Now we can state the main result.

\begin{theorem}[Optimal Defense Placement for Complete Graphs]
\label{thm:def_opt_reform}
    Given a complete digraph $\mathcal{G} = (\mathcal{V},\mathcal{E})$ with $|\mathcal{V}|=n$ and a defense budget $B \in \natural \cap (n,n^2)$, any permutation of
    \begin{equation} \label{eq:def_opt_soln}
    \begin{aligned}
        \tau_1 = \cdots = \tau_r &= \lceil B/n \rceil, \\
        \tau_{r+1} = \cdots = \tau_n &= \lfloor B/n \rfloor,
    \end{aligned}
    \end{equation}
    where $r = \textup{mod}(B,n)$ is a global minimum of problem~\eqref{eq:def_opt_reform}.
\end{theorem}

\begin{proof}
We begin with two comments. First, as discussed in Section~\ref{sect:opt_complete}, constraint~\eqref{const:nonl} uniquely defines $w$ as a function of $\tauvec$. 
% The remaining constraints confine $\tauvec$ to integral points within a simplex. 
Second, note that~\eqref{eq:def_opt_reform} is invariant with respect to permutation of the entries of $\tauvec$. 
% Therefore in~\eqref{eq:def_opt_soln}, the choice of which nodes to allocate $\lceil B/n \rceil$ to and which nodes to allocate $\lfloor B/n \rfloor$ to is arbitrary. 
Any permutation of~\eqref{eq:def_opt_soln} will achieve an identical value $w^*$ which we will show is optimal for problem~\eqref{eq:def_opt_reform}.

To prove optimality, we will present a simple algorithm that starts at any feasible point for $\tauvec$ and maintains feasibility while decreasing $w$ at each iteration. Regardless of the choice of initial point, the algorithm will converge to a permutation of~\eqref{eq:def_opt_soln}.
% and thereby prove optimality. 
The algorithm consists entirely of identifying two values $\tau_i,\tau_j$ of the current allocation such that $\tau_i \geq \tau_j + 2$ and updating their values to $\tau_i-1,\tau_j+1$, respectively. Clearly, this ``pairwise balancing" algorithm will maintain feasibility and converge to a permutation of~\eqref{eq:def_opt_soln}. What remains is to show that $w$ decreases at each iteration. 

Consider the function $f(w,\tauvec) := \sum_{i=1}^n w^{1/\tau_i}$ defined by the LHS of constraint~\eqref{const:nonl}. It can be seen that $f$ is increasing in $w$. Therefore, choosing a feasible $\tauvec$ that maximizes $f$ leads to the minimum value of $w$ according to constraint~\eqref{const:nonl}. Now fix $\tau_k$ for all $k \in \{1,\dots,n\}\backslash\{i,j\}$, and define the reduced function $f_{\textup{red}}(w,\tau_i,\tau_j) := w^{1/\tau_i} + w^{1/\tau_j}$ where $\tau_i \geq \tau_j + 2$. By the same argument, choosing $\tau_i, \tau_j$ to increase $f_{\textup{red}}$ leads to a decrease in $w$. Therefore, we need to show that the pairwise balancing algorithm increases $f_{\textup{red}}$ at each iteration, i.e., $f_{\textup{red}}(w,\tau_i,\tau_j) < f_{\textup{red}}(w,\tau_i-1,\tau_j+1)$. This is equivalent to showing that 
\begin{equation}
    w^{1/\tau_i} - w^{1/{(\tau_i-1)}} < w^{1/{(\tau_j+1)}} - w^{1/\tau_j}.    
\end{equation}
Because $\tau_i - 1 > \tau_j$, this condition can be interpreted as a ``diminishing returns" property of the function $g(w, \tau) = w^{1/\tau}$ w.r.t.\ $\tau$. We can show that $g$ has this property by establishing that $g$ is concave in $\tau$. Performing the calculation, we have:
\begin{equation}
\begin{aligned}
    \frac{\partial g}{\partial \tau} & = \frac{-w^{1/\tau}\ln{w}}{\tau^2} > 0, \\
        \frac{\partial^2 g}{\partial \tau^2} & = \frac{w^{1/\tau}\ln{w}(\ln{w} + 2\tau)}{\tau^4} < 0.
\end{aligned}
\end{equation}
where we have used the lower bound on $w^*$ from Lemma~\ref{lem:w_bound} as a lower bound on $w$ to conclude that the term $\ln{w} + 2\tau > 0$ for $\tau \geq 1$. Therefore, we have shown that each iteration of the pairwise balancing algorithm leads to a decrease in $w$ which concludes the proof. \qedhere
\end{proof}

\subsection{Complete Bipartite Graphs}
In this section, we present the optimal defense placement rule for complete bipartite graphs assuming a surveillance strategy of the form of Eq.~\eqref{eq:bipart_strat}. Lemma~\ref{lem:bipart_strat} gives the capture probability for these strategies. As before, we write an optimization problem to maximize the capture probability via allocation of defenses while utilizing the optimized patrol strategy given by Eq.~\eqref{eq:v_opt_bipart}:
\begin{subequations} \label{eq:def_opt_bipart}
\begin{align}
    \min_{w_p \in \real,w_q \in \real,\tauvec \in \natural^n} \quad & \max \{ w_p,w_q  \} \\
    \text{s.t.} \quad & \sum\nolimits_{i=1}^{n_p} w_p^{1/ \lfloor \tau_i^p / 2 \rfloor} = n_p-1, \label{const:bipart_nonl1} \\
    \quad & \sum\nolimits_{j=1}^{n_q} w_q^{1/ \lfloor \tau_j^q / 2 \rfloor} = n_q-1, \label{const:bipart_nonl2} \\
    \quad & \1_{n_p}^\top \tauvec^p + \1_{n_q}^\top \tauvec^q = B.
\end{align}
\end{subequations}
Note that $2(n_p + n_q) < B < 2(n_p^2 + n_q^2)$ for a nontrivial problem. Because $w_p,w_q$ only decrease with successive even values for $\tau_i^q, \tau_j^p$, we can restrict our choice of $\tau_i^q, \tau_j^p$ to the even natural numbers ($B$ also assumed even). 

Now consider dividing the overall budget $B$ into two ``sub-budgets" $B_p,B_q$ such that $B_p,B_q$ are even and $\sum_{i=1}^{n_p} \tau_i^p = B_p, \sum_{j=1}^{n_q} \tau_j^q = B_q, B_p + B_q = B$. After the choice of $B_p,B_q$, \eqref{eq:def_opt_bipart} becomes decoupled into the max of two instances of a problem similar to the complete graph problem \eqref{eq:def_opt_reform}:
\begin{subequations} \label{eq:def_opt_bipart_reform}
\begin{align}
    & \min_{w_p \in \real,\tauvec^p \in \natural^{n_p}} \quad w_p \\
    \text{s.t.} & \quad \sum_{i=1}^{n_p} w_p^{2/\tau_i^p} = n_p - 1, \label{const:nonl_bipart} \\
    & \quad \1_{n_p}^\top \tauvec^p = B_p, \\
    & \quad \tauvec^p \in \{2,4,\ldots \}^{n_p}.
\end{align}
\end{subequations}
We will show that a uniform allocation rule similar to that presented in Theorem~\ref{thm:def_opt_reform} is optimal for problem~\eqref{eq:def_opt_bipart_reform}. 
\begin{theorem}[Optimal Defense Placement for Complete Bipartite Graphs]
\label{thm:def_opt_bipart}
    Given a complete bipartite digraph $\mathcal{G} = (\mathcal{P},\mathcal{Q},\mathcal{E})$ with $|\mathcal{P}| = n_p$ and an even-valued defense sub-budget $B_p \in (2n_p,2n_p^2)$, any permutation of
    \begin{equation} \label{eq:def_bipart_soln}
    \begin{aligned}
        \tau_{1}^p = \cdots = \tau_s^p & = \begin{cases}
            \lceil B_p/n_p \rceil {+} 1, & \lceil B_p/n_p \rceil \textup{ odd},\\
            \lceil B_p/n_p \rceil, & \lceil B_p/n_p \rceil \textup{ even},
        \end{cases}\\
        \tau_{s{+}1}^p {=} \cdots = \tau_{n_p}^p & = \begin{cases}
            \lfloor B_p/n_p \rfloor, & \lfloor B_p/n_p \rfloor \textup{ even},\\
            \lfloor B_p/n_p \rfloor {-} 1, & \lfloor B_p/n_p \rfloor \textup{ odd},
        \end{cases}
    \end{aligned}   
    \end{equation}
    where $s \in \natural$ is the solution of $s\tau_1^p + (n_p-s)\tau_{n_p}^p = B_p$ is a global minimum of problem~\eqref{eq:def_opt_bipart_reform}.
\end{theorem}
\begin{proof}
    \begin{TAC}
        The proof follows the same steps as that of Theorem~\ref{thm:def_opt_reform} and is included in the arXiv version of this brief~\cite{YJ-GDG-XD-JRM-FB:23}. \qedhere
    \end{TAC}
    
    \begin{arXiv}
        As before, constraint~\eqref{const:nonl_bipart} uniquely defines $w_p$ as a function of $\tauvec^p$. 
        % The remaining constraints confine $\tauvec^p$ to even-valued points within a simplex. 
        Once again,~\eqref{eq:def_opt_bipart_reform} is also invariant with respect to permutation of the entries of $\tauvec^p$. Any permutation of~\eqref{eq:def_bipart_soln} will achieve an identical value $w_p^*$ which we will show is optimal for problem~\eqref{eq:def_opt_bipart_reform}.

        To prove optimality, we will present a simple algorithm that starts at any feasible point for $\tauvec^p$ and maintains feasibility while decreasing $w_p$ at each iteration. Regardless of the choice of initial point, the algorithm will converge to a permutation of~\eqref{eq:def_bipart_soln} and thereby prove optimality. The algorithm consists entirely of identifying two values $\tau_i^p,\tau_j^p$ of the current allocation such that $\tau_i \geq \tau_j + 4$ and updating their values to $\tau_i-2,\tau_j+2$, respectively. Clearly, this ``pairwise balancing" algorithm will maintain feasibility and converge to a permutation of~\eqref{eq:def_bipart_soln}. What remains is to show that $w_p$ decreases at each iteration. 
        
        Consider the function $f(w_p,\tauvec^p) := \sum_{i=1}^n w_p^{2/\tau_i^p}$ defined by the LHS of constraint~\eqref{const:nonl_bipart}. It can be seen that $f$ is increasing in $w_p$. Therefore, choosing a feasible $\tauvec^p$ that maximizes $f$ leads to the minimum value of $w_p$ according to constraint~\eqref{const:nonl_bipart}. Now fix $\tau_k^p$ for all $k \in \{1,\dots,n\}\backslash\{i,j\}$, and define the reduced function $f_{\textup{red}}(w_p,\tau_i^p,\tau_j^p) := w_p^{2/\tau_i^p} + w^{2/\tau_j^p}$ where $\tau_i^p \geq \tau_j^p + 4$. By the same argument, choosing $\tau_i^p, \tau_j^p$ to increase $f_{\textup{red}}$ leads to a decrease in $w_p$. Therefore, we need to show that the pairwise balancing algorithm increases $f_{\textup{red}}$ at each iteration, i.e., $f_{\textup{red}}(w_p,\tau_i^p,\tau_j^p) < f_{\textup{red}}(w_p,\tau_i^p-2,\tau_j^p+2)$. This is equivalent to showing that 
        \begin{equation}
            w^{2/\tau_i^p} - w^{2/{(\tau_i^p-2)}} < w^{2/{(\tau_j^p+2)}} - w^{2/\tau_j^p}.    
        \end{equation}
        Because $\tau_i^p - 2 > \tau_j^p$, this condition can be interpreted as a ``diminishing returns" property of the function $g(w_p,\tau^p) = w_p^{2/\tau^p}$. We can show that $g$ has this property by establishing that the derivative of $g$ w.r.t.\ $\tau^p$ is decreasing, i.e., that $g$ is concave in $\tau^p$. Performing the calculation, we arrive at:
        \begin{equation}
        \begin{aligned}
            \frac{\partial g}{\partial \tau^p} & = \frac{-2w_p^{2/\tau^p}\ln{w_p}}{(\tau^p)^2} > 0, \\
            \frac{\partial^2 g}{(\partial \tau^p)^2} & = \frac{4w_p^{2/\tau^p}\ln{w_p}(\ln{w_p} + \tau^p)}{(\tau^p)^4}.
        \end{aligned}
        \end{equation}
        To complete the proof, we need to show that $\ln{w_p}+\tau^p > 0$. For $\tau^p \geq 2$, this is equivalent to showing that $w_p > \me^{-2}$. 
        
        We proceed by contradiction. Assume $w_p \leq \me^{-2}$. We will show that this cannot satisfy constraint~\eqref{const:nonl_bipart} for any feasible choice of $\tauvec^p$. Because $f$ is increasing in $w_p$, it suffices to show that $f(\me^{-2},\tauvec^p) < n_p - 1$ for any feasible $\tauvec^p$. It is easy to see that $f(\me^{-2},\tauvec^p)$ is strictly increasing and concave w.r.t.\ $\tau^p$:
        \begin{equation}
        \begin{aligned}
            \frac{\partial f(\me^{-2},\tauvec^p)}{\partial \tau_i^p} & = \frac{4\me^{-4/\tau_i^p}}{(\tau_i^p)^2} > 0, \\
            \frac{\partial^2 f(\me^{-2},\tauvec^p)}{(\partial \tau_i^p)^2} & = \frac{-8\me^{-4/\tau_i^p}(\tau_i^p-2)}{(\tau_i^p)^4} \leq 0.
        \end{aligned}
        \end{equation}
        Therefore, the following optimization problem is convex:
        \begin{equation} \label{opt:lemma_bipart}
        \begin{aligned}
            \max_{\tauvec^p \in \real^{n_p}} \quad & f(\me^{-2},\tauvec^p) \\
            \text{s.t.} \quad & \tau_i^p \geq 2, \ \forall i \\
            & \1_{n_p}^\top \tauvec^p = B_p.
        \end{aligned}
        \end{equation}
        Notice that we have relaxed $\tauvec^p$ to be real-valued. The optimal value of~\eqref{opt:lemma_bipart} represents an upper bound on $f(\me^{-2},\tauvec^p)$ for feasible $\tauvec^p$. We form the Lagrangian of~\eqref{opt:lemma_bipart}:
        \begin{equation}
            \mathcal{L}(\tauvec^p,\lambda) = f(\me^{-2},\tauvec^p) + \lambda(\1_{n_p}^\top \tauvec^p - B_p),
        \end{equation}
        where we have again increased the size of the feasible region by omitting the inequality constraints. Differentiating w.r.t.\ $\tau_i^p$ yields the following condition for stationary points:
        \begin{equation} \label{eq:bipart_stat_pt}
            \frac{\me^{-4/\tau_i^p}}{(\tau_i^p)^2} = \frac{\me^{-4/\tau_j^p}}{(\tau_j^p)^2},
        \end{equation}
        for all $i,j$. The point $(\tauvec^p)^* = (B_p/n_p)\1_{n_p}$ satisfies~\eqref{eq:bipart_stat_pt} and the constraints and therefore must be a global maximum of problem~\eqref{opt:lemma_bipart}.
        Now we have an upper bound for $f(\me^{-2},\tauvec^p)$; namely, $f(\me^{-2},\tauvec^p) \leq f(\me^{-2},(\tauvec^p)^*) = n_p\me^{-4n_p/B_p}$. Because $B_p < 2n_p^2$, $n_p\me^{-4n_p/B_p} < n_p\me^{-2/n_p}$. It can be shown that $n_p\me^{-2/n_p} < n_p-1$ for $n_p \geq 2$ which is a contradiction.
    \end{arXiv}
\end{proof}

Given Theorem~\ref{thm:def_opt_bipart}, solving \eqref{eq:def_opt_bipart} boils down to appropriately dividing the overall budget $B$ into two even sub-budgets $B_p,B_q$. Note that $w_p$ is a decreasing functions of $B_p$ and $w_q$ is an increasing function of $B_p$ because $B_q = B - B_p$. Thus, we propose Algorithm \ref{alg:mod_bisection}, a modified bisection algorithm on $B_p$, to determine the optimal sub-budgets $B_p^*,B_q^*$.
\begin{algorithm}\small
\caption{Modified bisection}\label{alg:mod_bisection}
$LB \gets 2n_p,\ UB \gets B - 2n_q$\;
\While{$UB - LB > 2$}{
    $B_p \gets (LB + UB)/2$\;
    \lIf{$B_p\ \textup{odd}$}{$B_p \gets B_p + 1$}
    $B_q \gets B - B_p$\;
    Allocate $\tau_i^p,\tau_j^q$ using~\eqref{eq:def_bipart_soln}\;
    Solve (\ref{const:bipart_nonl1}),(\ref{const:bipart_nonl2}) for $w_p,w_q$ via bisection\;
    \eIf{$w_p < w_q$}{
        $UB \gets B_p$\;
    }{
        $LB \gets B_p$\;
    }
}
\end{algorithm}

The initial values for $LB,UB$ in Algorithm \ref{alg:mod_bisection} ensure that all nodes have an attack duration of at least two which ensures a nonzero capture probability. The final values for $LB,UB$ are candidates for the optimal value of $B_p$ and must be checked individually to identify the overall optimum. 
% There are three possible cases for $w_p,w_q$ as $B_p$ varies from $2n_p$ to $B-2n_q$:
% \begin{enumerate}
%     \item $w_p < w_q$ for all values of $B_p$,
%     \item $w_p > w_q$ for all values of $B_p$,
%     \item there is a critical value $\hat{B}_p$ such that $w_p(\hat{B}_p) > w_q(\hat{B}_p)$ and $w_p(\hat{B}_p + 2) \leq w_q(\hat{B}_p + 2)$.
% \end{enumerate}
% In case 1, the while loop in Algorithm \ref{alg:mod_bisection} terminates with $LB = B - 2n_q - 2, UB = B - 2n_q$. In case 2, the while loop terminates with $LB = 2n_p, UB =2n_p + 2$. In case 3, the while loop terminates with $LB = \hat{B}_p, UB =\hat{B}_p + 2$. In all cases, checking these two values identifies the optimal $B_p^*$. 

\begin{figure}[!htb]
    \centering
    \includegraphics[width=0.45\textwidth]{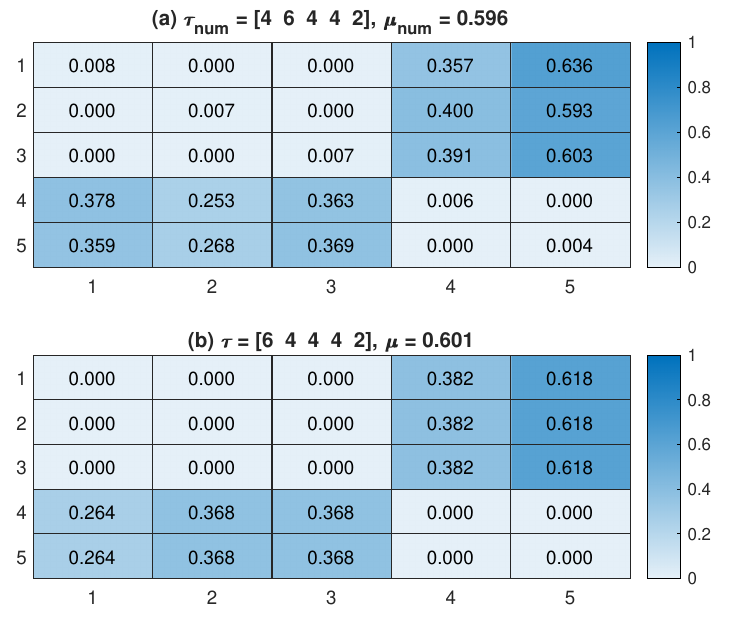}
    \caption{We consider a complete bipartite graph where $n_p = 3, n_q = 2$ and the defense budget $B = 20$. Optimized defense allocations $\tau$, transition matrices, and resulting capture probabilities $\mu$ come from (a) \texttt{fmincon} and (b) the proposed method.}
    \label{fig:num_ex}
\end{figure}

\section{Numerical Example}
Fig.~\ref{fig:num_ex}(b) shows a numerical example to illustrate the defense placement and patrol strategy computation. 
Algorithm~\ref{alg:mod_bisection} identifies $B_p = 14, B_q = 6$ as the optimal sub-budgets yielding $\tauvec^p = \begin{bmatrix} 6 & 4 & 4 \end{bmatrix}^\top, \tauvec^q = \begin{bmatrix} 4 & 2 \end{bmatrix}^\top$ as the optimal defense allocations. 
% For this small example, checking all even integers between 6 and 16 for $B_p$ verifies that this choice is indeed optimal. 
% In this case, we can verify via enumeration that the resulting defense placement is optimal. 
Note that this defense allocation leads to a $4.5\%$ higher capture probability than the uniform allocation $\tauvec^p = \begin{bmatrix} 4 & 4 & 4 \end{bmatrix}^\top, \tauvec^q = \begin{bmatrix} 4 & 4 \end{bmatrix}^\top$. % Solving Eq.~\eqref{eq:v_opt_bipart} with the values for $\tauvec^p,\tauvec^q$ from Algorithm~\ref{alg:mod_bisection} yields the patrol strategy and capture probability $\mu$ shown in Fig.~\ref{fig:num_ex}(a). 
The best defense allocation and transition matrix combination found by MATLAB's \texttt{fmincon} among 1000 random initializations is shown in Fig.~\ref{fig:num_ex}(a) and closely resembles our result.

\section{Conclusions}
In this brief, we generalized the Stackelberg game formulation for stochastic surveillance to accommodate heterogeneous defenses across the graph. We found an upper bound for the surveillance agent's probability of capturing the attacker. We presented methods for computing effective patrol strategies along with corresponding suboptimality bounds for complete graphs and complete bipartite graphs. We proved that the proposed patrol strategy for complete bipartite graphs is optimal in the special case of star graphs. Additionally, we identified the optimal defense allocations corresponding to our proposed patrol strategies for complete and complete bipartite graphs. 
% Finally, we remark that more general theoretical results are difficult to obtain due to the complexity of the formulation. 
In future, we will explore avenues for desirable extensions such as arbitrary graph topology and travel times on edges. 


%%%%%%%%%%%%%%%%%%%%%%%%%%%%%%%%%%%%%%%%%%%%%%%%%%%%%%%%%%%%%%%%%%%%%%%%%%%%%%%%

\begin{arXiv}
    \section*{Appendix}
    \subsection*{Constant Factor Suboptimal Strategy on Complete Bipartite Graphs with Uniform Attack Duration}
    In \cite{XD-DP-FB:19b}, the optimal strategy for star graphs is presented in the case of uniform attack duration. Complete bipartite graphs can be viewed as a generalization of star graphs. Note that $n_q \geq 2, n_p \geq 2, n \geq 4$ because if either set has cardinality 1, then the graph is a star. Also note that $2 \leq \tau \leq 2n-4$ for a nontrivial game. Motivated by the results in Section \ref{sect:opt_bipart}, consider the following strategy:
    \begin{equation} \label{eq:P_bipart}
        P = \begin{bmatrix}
                \0_{n_q \times n_q} & (1/n_p)\1_{n_q \times n_p} \\
                (1/n_q)\1_{n_p \times n_q} & \0_{n_p \times n_p} \\
            \end{bmatrix}.
    \end{equation}
    The corresponding stationary distribution is
    \begin{equation}
        \pivec^\top = 
        \begin{bmatrix}
            \Bigl(\frac{1}{2n_q}\Bigr)\1_{1\times n_q} & \Bigl(\frac{1}{2n_p}\Bigr)\1_{1\times n_p}
        \end{bmatrix}.
    \end{equation}
    Using the recursion~\eqref{eq:F_k}, it can be shown that the capture probability $\mu$ of this strategy is
    \begin{equation} \label{eq:mu_bipart}
        \mu = \min \Biggl\{ 1-\biggl(1-\frac{1}{n_q}\biggr)^{\lfloor \tau/2 \rfloor}, 1-\biggl(1-\frac{1}{n_p}\biggr)^{\lfloor \tau/2 \rfloor} \Biggr\}
    \end{equation}
    where the minimum is attained by the larger of $n_q,n_p$. Using this expression, we can derive a constant factor suboptimality bound that is precisely half the bound derived in \cite{XD-DP-FB:19b} for the analogous complete graph strategy.
    \begin{theorem}[Constant Factor Suboptimal Strategy for Complete Bipartite Graphs]
        Given a complete bipartite digraph $\mathcal{G} = (\mathcal{P},\mathcal{Q},\mathcal{E})$ with $|\mathcal{P}| = n_q,|\mathcal{Q}| = n_p$ and an attack duration $\tau \in \natural$, the following inequalities hold:
        \begin{equation}
            \begin{cases}
                \frac{\mu}{\mu^*} \geq \frac{1}{3}, \quad & \tau \textup{ odd} \\
                \frac{\mu}{\mu^*} \geq \frac{1}{2}(1-\frac{1}{e}), \quad & \tau \textup{ even}
            \end{cases}
        \end{equation} 
        where $\mu$ is the capture probability corresponding to the strategy given in Eq.~\eqref{eq:P_bipart} and $\mu^*$ is the optimal capture probability.
    \end{theorem}
    \begin{proof}
    First we remove the dependency in Eq.~\eqref{eq:mu_bipart} on $n_q,n_p$. Because the function $f(x,t) = 1 - (1 - 1/x)^t$ is decreasing in $x$ for $x \geq 2, t \in \natural$, we have the following lower bound for $\mu$:
    \begin{equation} \label{eq:mu_bipart_bd}
        \mu \geq 1-\biggl(1-\frac{1}{n-2}\biggr)^{\lfloor \tau/2 \rfloor}.
    \end{equation}
    Now we consider separately both cases for the parity of $\tau$.
    
    \textbf{$\tau$ odd}: 
    Consider the ratio of the capture probability bound in Eq.~\eqref{eq:mu_bipart_bd} to the optimal capture probability bound $\mu^* \leq \tau/n$ from \cite{XD-DP-FB:19b}:
    \begin{equation}
        \frac{\mu}{\mu^*} \geq \frac{1-\bigl(1-\frac{1}{n-2}\bigr)^{(\tau - 1)/2}}{\tau / n}
    \end{equation}
    where $\tau \in \{2k+1: k\in \{1,\ldots,n-3 \} \}$. Using the inequality $(1 + x)^r \leq 1/(1-rx)$ for $r\geq 0, x\in [-1,1/r]$, we have the following:
    \begin{equation}
        \frac{\mu}{\mu^*} \geq \frac{n(\tau - 1)}{\tau(2n+\tau - 5)} := h(n,\tau).
    \end{equation}
    Differentiating $h(n,\tau)$ w.r.t.\ $n$ yields
    \begin{equation}
        \frac{\partial h}{\partial n} = \frac{(\tau - 5)(\tau - 1)}{\tau(2n+\tau -5)}.
    \end{equation}
    Consider three cases:
    \begin{equation}
        \begin{cases}
            \frac{\partial h}{\partial n} < 0, \quad & \tau = 3 \\
            \frac{\partial h}{\partial n} = 0, \quad & \tau = 5 \\
            \frac{\partial h}{\partial n} > 0, \quad & \tau \in \{2k+1: k\in \{3,\ldots,n-3 \} \}.
        \end{cases}
    \end{equation}
    For $\tau = 3$ we have the following lower bound for $h(n,\tau)$:
    \begin{equation}
        h(n,3) \geq \lim_{n \to \infty} h(n,3) = \lim_{n \to \infty} \frac{n}{3n-3} = \frac{1}{3}.
    \end{equation}
    For $\tau = 5$ we have the following:
    \begin{equation}
        h(n,5) = h(5,5) = \frac{2}{5}.
    \end{equation}
    For $\tau \in \{2k+1: k\in \{3,\ldots,n-3 \} \}$ we have that $n \geq (\tau + 5)/2$ which implies the following:
    \begin{equation}
    \begin{aligned}
        h(n,\tau) & \geq h((\tau + 5)/2,\tau) = \frac{(\tau + 5)(\tau - 1)}{4\tau^2} \\
        & \geq h(6,7) = \frac{18}{49}.
    \end{aligned}
    \end{equation}
    where the second inequality arises because $h((\tau + 5)/2,\tau)$ is decreasing in $\tau$. Taking the smallest of the three lower bounds we have the desired result for the $\tau$ odd case: $\mu/\mu^* \geq 1/3$.
    
    \textbf{$\tau$ even}: 
    Now consider the ratio of the capture probability bound in Eq.~\eqref{eq:mu_bipart_bd} to the optimal capture probability bound $\mu^* \leq \tau/n$ from \cite{XD-DP-FB:19b}:
    \begin{equation}
        \frac{\mu}{\mu^*} \geq \frac{1-\bigl(1-\frac{1}{n-2}\bigr)^{\tau/2}}{\tau / n} := g(n,\tau)
    \end{equation}
    where $\tau \in \{ 2k: k\in \{1,\ldots,n-2 \} \}$. Differentiating $g(n,\tau)$ w.r.t.\ $\tau$ yields
    \begin{equation}
        \frac{\partial g}{\partial \tau} = \frac{n}{\tau^2} \Biggl[ \biggl(1 - \frac{\tau}{2} \ln{\biggl( \frac{n-3}{n-2} \biggr)} \biggr) \biggl(\frac{n-3}{n-2}\biggr)^{\tau/2} - 1\Biggr].
    \end{equation}
    Using the inequality $\ln{x} \geq 1 - 1/x$, we have
    \begin{equation}
        \frac{\partial g}{\partial \tau} \leq \frac{n}{\tau^2} \Biggl[ \biggl(1 + \frac{\tau}{2(n-3)} \biggr) \biggl(\frac{n-3}{n-2}\biggr)^{\tau/2} - 1\Biggr].
    \end{equation}
    Define the following:
    \begin{equation}
    \begin{aligned}
        p(n,\tau) & := \biggl(1 + \frac{\tau}{2(n-3)} \biggr) \biggl(\frac{n-3}{n-2}\biggr)^{\tau/2}, \\ 
        q(n,\tau) & := \ln{\bigl( p(n,\tau) \bigr)}
    \end{aligned}
    \end{equation}
    where we want to show that $\frac{\partial g}{\partial \tau} \leq 0$ by showing that $q(n,\tau) \leq 0$ and $p(n,\tau) \leq 1$ for $\{ \tau = 2k: k\in \{1,\ldots,n-2 \} \}$. Differentiating $q(n,\tau)$ w.r.t.\ $\tau$ yields
    \begin{equation}
    \begin{aligned}
        \frac{\partial q}{\partial \tau} & = \frac{1}{2n+\tau -6} + \frac{1}{2} \ln{\biggl(\frac{n-3}{n-2} \biggr)} \\
        & \leq \frac{1}{2n+\tau -6} - \frac{1}{2n-4} \leq 0
    \end{aligned}
    \end{equation}
    where we have used that $\ln{x} \leq x-1$ for the first inequality. Because $q(n,\tau)$ is decreasing in $\tau$, $p(n,\tau)$ is also decreasing in $\tau$. Note that $q(n,2) = 0$ and $p(n,2) = 1$. Therefore, $\frac{\partial g}{\partial \tau} \leq 0$ which implies that
    \begin{equation}
        g(n,\tau) \geq g(n,2n-4) = \frac{1-\bigl(1-\frac{1}{n-2}\bigr)^{n-2}}{(2n-4)/ n} := y(n).
    \end{equation}
    It can be shown that $y(n)$ is decreasing and thus we have the desired result for the $\tau$ even case:
    \begin{equation}
        \frac{\mu}{\mu^*} \geq g(n,\tau) \geq \lim_{n \to \infty} y(n) = \frac{1}{2} \biggl(1 - \frac{1}{e} \biggr).
    \end{equation}
    \end{proof}
\end{arXiv}

%%%%%%%%%%%%%%%%%%%%%%%%%%%%%%%%%%%%%%%%%%%%%%%%%%%%%%%%%%%%%%%%%%%%%%%%%%%%%%%%

\bibliographystyle{plain}
\bibliography{bib/alias,bib/FB,bib/Main,bib/YJ}

\end{document}